\documentclass[letterpaper, 10 pt, conference]{ieeeconf}

\IEEEoverridecommandlockouts
\usepackage[T1]{fontenc}

\usepackage{amsmath}
\usepackage{amssymb}

\usepackage{amsthm}
\usepackage{graphicx}
\usepackage[dvipsnames]{xcolor}
\usepackage[normalem]{ulem}
\usepackage{balance}

\newtheorem{theorem}{Theorem}
\newtheorem{definition}{Definition}

\newtheorem{proposition}{Proposition}

\newtheorem{assumption}{Assumption}
\usepackage{algorithm}
\usepackage{algpseudocode}

\title{\LARGE \bf
Minimal Information Control Invariance via Vector Quantization
}
\begin{document}

\newcommand{\teodor}[1]{\textcolor{ForestGreen}{#1}}
\newcommand{\sayan}[1]{\textcolor{blue}{#1}}

\author{Ege Yuceel$^{1}$, Teodor Tchalakov$^{1}$, and Sayan Mitra$^{1}$%
\thanks{$^{1}$Department of Electrical and Computer Engineering, University of Illinois at Urbana-Champaign, Urbana, IL 61801, USA. {\tt\small \{eyceel2, ttcha2, mitras\}@illinois.edu}}}

\maketitle

\thispagestyle{empty}
\pagestyle{empty}

%%%%%%%%%%%%%%%%%%%%%%%%%%%%%%%%%%%%%%%%%%%%%%%%%%%%%%%%%%%%%%%%%%%%%%%%%%%%%%%%
\begin{abstract}
Safety-critical autonomous systems must satisfy hard state constraints under tight computational and sensing budgets, yet learning-based controllers are often far more complex than safe operation requires. To formalize this gap, we study how many distinct control signals are needed to render a compact set forward invariant under sampled-data control, connecting the question to the information-theoretic notion of invariance entropy. We propose a vector-quantized autoencoder that jointly learns a state-space partition and a finite control codebook, and develop an iterative forward certification algorithm that uses Lipschitz-based reachable-set enclosures and sum-of-squares programming. On a 12-dimensional nonlinear quadrotor model, the learned controller achieves a $157\times$ reduction in codebook size over a uniform grid baseline while preserving invariance, and we empirically characterize the minimum sensing resolution compatible with safe operation.
\end{abstract}
\section{Introduction}

Autonomous systems such as quadrotors~\cite{mahony2012multirotor}, fixed-wing UAVs~\cite{beard2012small}, and embedded robotic platforms must satisfy hard safety constraints while operating under tight computational, sensing, and communication budgets. Controllers for such platforms often consume far more resources than may be necessary: vision-based obstacle avoidance on a nano-drone has been achieved by distilling a 3.1M-parameter depth network down to a 310K-parameter model~\cite{zhang2024end}, pruning techniques can reduce trained network parameter counts by over $90\%$ without loss of accuracy~\cite{frankle2018lottery,han2015learning}, and dynamic obstacle avoidance on quadrotors has been demonstrated using the sparse event stream of an event camera in place of dense video~\cite{falanga2020dynamic}. This suggests that the resources strictly necessary for safe operation may be far smaller than what current practice allocates. There are many notions of minimality one could pursue; in this paper we focus on one: what is the minimum number of distinct control signals needed to guarantee that the system remains inside a prescribed safe set?

Two threads of prior work address aspects of this question. In
control theory, invariance~\cite{colonius2009invariance, kawan2013invariance} and stabilization~\cite{colonius2012stabilization} entropy characterize the minimum data rate necessary over the control channel for
achieving their corresponding performance. These results provide fundamental lower bounds but  do not yield constructive procedures for designing controllers that
achieve or even approach these lower bounds. One such controller for LTI systems is~\cite{hespanha2002towards}, which proposes a practical encoding/decoding scheme. On the computer vision side, reducing sensor resolution directly lowers power consumption and onboard computation. Recent work shows that even extremely low-resolution sensors can solve vision-based control tasks when their design is optimized computationally~\cite{atanov2024far, klotz2024minimalist}, but these efforts treat sensing in isolation and do not provide guarantees about safety or invariance. When the control channel also operates at a minimal bitrate, the sensor need only distinguish among fewer actions, which can relax the sensing requirements in turn.

\begin{figure}[t]
    \centering
    \includegraphics[width=0.97\linewidth]{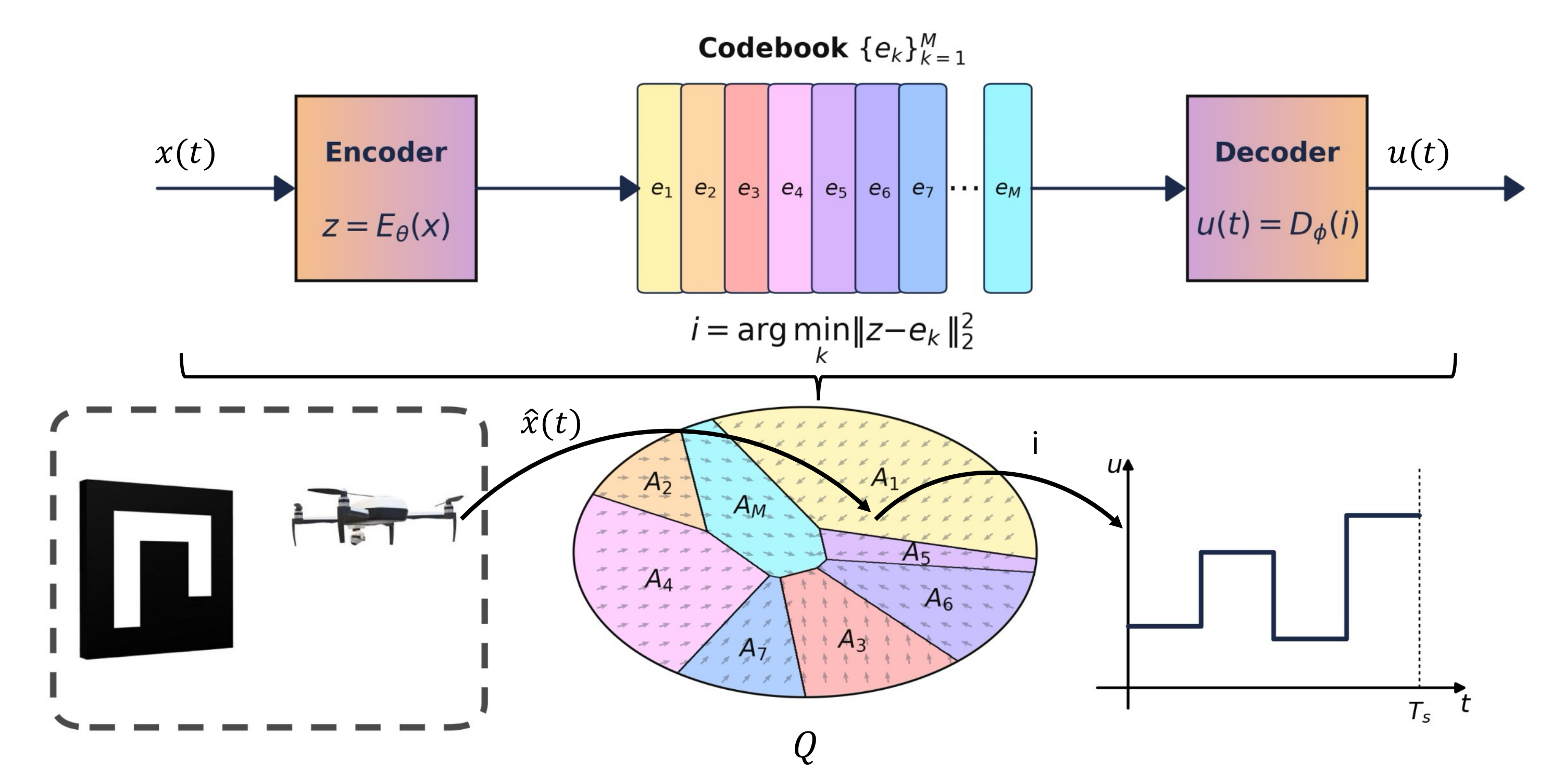}
    \caption{VQ-AE controller architecture. The quadrotor estimates its state $\hat{x}(t) \in \mathbb{R}^{12}$ from images of an ArUco marker. The encoder $E_\theta$ maps $\hat{x}(t)$ to a latent vector $z$, quantized to the nearest codebook entry $i = \arg\min_k \|z - e_k\|_2^2$, inducing a partition $\{A_i\}_{i=1}^M$ of $Q$. A 2D projection of the partition is shown; the full 12-dimensional partition contains additional regions. Arrows display the controlled vector field within each region, directing trajectories inward to maintain invariance of $Q$. The decoder $D_\phi$ maps index $i$ to an open-loop control applied over $[t_k, t_k + T_s)$.}
    \label{fig:}
\end{figure}

In this paper, we study the following problem: Given a nonlinear
dynamical system $\dot{x} = f(x, u)$
and a candidate invariant set $Q \subset \mathbb{R}^n$, what is the {\em minimum number\/} of distinct open-loop control signals $u(\cdot)$ that is sufficient to render $Q$ forward invariant
under sampled-data control for a given period $T_s$? We define the \emph{complexity} of a controller as the number of distinct open-loop control signals it uses. Minimizing this complexity is therefore practically meaningful: a simpler controller requires fewer distinct signals to be distinguished, stored, and transmitted, reducing demands on sensors, communication networks, and onboard computation. Sufficient conditions for the existence of minimal controllers have been established~\cite{colonius2009invariance, kawan2013invariance} through, for example, grid partitions and Lipschitz-type arguments; related minimum data-rate problems for estimation have been analyzed via matrix measure-based entropy bounds~\cite{liberzon2016entropy, sibai2018state}.

Concretely, we present a method that connects invariance entropy theory to constructive learning-based synthesis in three steps. (1) We learn a compact partition of the state space and an associated control codebook via a vector quantized autoencoder (VQ-AE), (2) we certify the learned controller through an iterative forward certification (IFC) algorithm, and (3) we empirically characterize the sensing resolution compatible with safe operation. The resulting controller achieves invariance with orders-of-magnitude fewer control signals.

The encoder $E_\theta: \mathbb{R}^n \to \mathcal{Z}$ maps each state to a discrete index in a finite codebook $\mathcal{Z} = \{1,\dots,M\}$, inducing a partition on $Q$. The decoder $D_\phi: \mathcal{Z} \to \mathcal{U}_{T_s}$ assigns each region a fixed open-loop control sequence over the sampling horizon $T_s$. We design the training objective to balance four competing goals: an invariance loss penalizes trajectories that leave $Q$, a codebook entropy loss regularizes the number of active codes~\cite{wang2026theoretically}, a commitment loss anchors encoder outputs to their assigned codebook vectors via the straight-through estimator~\cite{bengio2013estimating, van2017neural}, and a control effort loss prevents excessive control signals. The result is a compact codebook of open-loop control sequences sufficient to keep the system within $Q$.

While data-driven controllers have demonstrated strong empirical performance in complex control tasks, formal safety guarantees are rare~\cite{brunke2022safe}, and to the best of our knowledge no VQ-based representation carries a certificate of set invariance. We address this gap by imposing structure on the encoder while leaving the decoder unrestricted: when the encoder is an invertible affine map, nearest-neighbor code assignment induces a convex partitioning in state space, enabling region-wise verification. The IFC algorithm propagates polyhedral reachable-set enclosures through each control segment using Lipschitz-based bloating and certifies invariance at each step via a sum-of-squares (SOS) feasibility program~\cite{parrilo2003semidefinite}. We prove that if the algorithm terminates successfully, the closed-loop system is forward invariant on $Q$, and demonstrate certification on the 12-dimensional quadrotor.

Finally, we connect control minimality to sensing. Since the controller selects among finitely many distinct controls, the observer need only determine at the sampling period $T_s$ which partition region the state belongs to, not reconstruct the full state. This relaxation suggests that the sensing resolution required for invariance is governed by the partition granularity. We empirically identify the minimum monocular camera resolution sufficient to support invariance on the quadrotor by training an observer network that estimates the state from images of an ArUco marker~\cite{garrido2014automatic} at progressively coarser resolutions and recording the threshold at which closed-loop safety degrades.

The remainder of the paper is organized as follows.
Section~II introduces the problem formally and defines the invariance
objective. Section~III reviews invariance entropy and spanning sets.
Section~IV presents the VQ-AE training procedure and controller design.
Section~V develops the iterative forward certification algorithm and
its correctness proof. Section~VI reports empirical results on a
12-dimensional nonlinear quadrotor, including baseline comparisons,
codebook compression, and the minimum sensing resolution compatible
with invariance.

\section{Problem Definition}
Consider the nonlinear continuous-time system
\begin{equation} 
\dot{x}(t) = f(x(t),u(t)),
\label{eq:plant}
\end{equation}
where $x(t) \in \mathbb{R}^n$, $u(t) \in \mathcal{U} \subset \mathbb{R}^m$,
and $f:\mathbb{R}^n \times \mathcal{U} \to \mathbb{R}^n$ is Lipschitz in $x$.
Let $Q \subset \mathbb{R}^n$ be a compact set with nonempty interior.
The objective is to design a sampled-data controller with sampling
period $T_s > 0$ such that
\begin{equation}
x(0)\in Q \;\Rightarrow\; x(t)\in Q \quad \forall\, t\ge 0,
\label{eq:invariance}
\end{equation}
using the minimum number of distinct open-loop control signals.

\section{Preliminaries}
\textbf{Notation.}
For $T>0$, let $\mathcal{U}_T := L^\infty([0,T],\mathcal{U})$ denote
the space of admissible open-loop control signals of duration $T$, and
let $\|\cdot\|$ denote the Euclidean norm. For an initial condition
$x_0 \in \mathbb{R}^n$ and input $u(\cdot) \in \mathcal{U}_T$, let
$\varphi(t, x_0, u)$ denote the unique solution of~\eqref{eq:plant} at
time $t \ge 0$, and write $\varphi(t, A, u) := \{\varphi(t, x_0, u) :
x_0 \in A\}$ for a bounded set $A \subset \mathbb{R}^n$. The Lipschitz constant
of $f(\cdot, u)$ over a bounded set $A$ for fixed $u$ is denoted
$L(f(\cdot, u), A)$.

\begin{assumption}[Weak Invariance]
\label{ass:weak_invariance}
For every $x \in Q$, there exists $u_x \in \mathcal{U}_{T_s}$ such that
\[
\varphi(t, x, u_x) \in Q \quad \forall\, t \in [0, T_s].
\]
\end{assumption}

Assumption~1 is necessary because without a feasible control for every state in $Q$, no controller can render $Q$ invariant, even with infinitely many open-loop control signals.

\textbf{Minimal Controllers and Invariance Entropy.}
Fix $T>0$. By Assumption~\ref{ass:weak_invariance}, every state
$x\in Q$ admits at least one control $u\in\mathcal{U}_T$ keeping it
inside $Q$ over $[0,T]$. The question is whether a \emph{finite}
collection of such signals suffices: if $Q$ can be covered by regions,
each assigned a single control that keeps every state in that region
inside $Q$ over the horizon $T$, then invariance is achievable with
a finite set of open-loop controls \cite{colonius2009invariance}.
\begin{definition}[$T$-Spanning Set]
Let $T>0$.
A finite set $\mathcal V_T \subset \mathcal U_T$
is called \emph{$T$-spanning} for $Q$
if there exists an open cover $\mathcal A$ of $Q$
and a selection map
$\alpha:\mathcal A \to \mathcal V_T$
such that
\[
\varphi(t,A,\alpha(A)) \subset Q
\quad \forall A\in\mathcal A,\ \forall t\in[0,T],
\]
where
\[
\varphi(t,A,u)
:=
\{\varphi(t,x,u):x\in A\}.
\]
The minimal cardinality of such a family is
\[
r_{\mathrm{inv}}(T,Q)
:=
\min
\bigl\{
|\mathcal V_T| :
\mathcal V_T \text{ is $T$-spanning for } Q
\bigr\}.
\]
\end{definition}
The quantity $r_{\mathrm{inv}}(T,Q)$ measures
the minimal number of open-loop control signals
of duration $T$ required to guarantee invariance of $Q$.
As $T$ increases, this number may grow.
Its asymptotic growth rate is captured by the following definition from~\cite{kawan2013invariance}.
\begin{definition}[Invariance Entropy]
The asymptotic growth rate of $r_{\mathrm{inv}}(T,Q)$ is called the
\emph{invariance entropy} of the system on $Q$ and is defined by
\[
h_{\mathrm{inv}}(Q)
:=
\limsup_{T\to\infty}
\frac{1}{T}
\log r_{\mathrm{inv}}(T,Q).
\]
\end{definition}
\begin{proposition}
\label{prop:finite_horizon_inv_entropy}
\[
h_{\mathrm{inv}}(Q)
=
\inf_{T>0}
\frac{1}{T}
\log r_{\mathrm{inv}}(T,Q).
\]
\end{proposition}

\begin{proof}
Concatenation of $T_1$- and $T_2$-spanning sets yields a
$(T_1+T_2)$-spanning set of cardinality at most
$r_{\mathrm{inv}}(T_1,Q)\,r_{\mathrm{inv}}(T_2,Q)$, so
$\log r_{\mathrm{inv}}(T,Q)$ is subadditive. The result follows
from Fekete's lemma~\cite{colonius2009invariance,kawan2013invariance}.
\end{proof}

Proposition~\ref{prop:finite_horizon_inv_entropy} implies $r_{\mathrm{inv}}(T_s,Q) \ge e^{h_{\mathrm{inv}}(Q)\,T_s}$ for any $T_s > 0$: a larger sampling period demands exponentially more distinct control signals, and no controller can maintain invariance with fewer than $e^{h_{\mathrm{inv}}(Q)\,T_s}$ codes regardless of architecture. In practice, $T_s$ is fixed by the sampling rate, so we focus on control sets defined over this horizon.
\section{Minimal Controller Design via VQ-AEs}

Fix the sampling period $T_s$ and consider a minimal
$T_s$-spanning set
\[
\mathcal V_{T_s}
=
\{u_1,\dots,u_M\}
\subset \mathcal U_{T_s},
\qquad
M = r_{\mathrm{inv}}(T_s,Q).
\]
By definition, there exists an open cover
$\mathcal A=\{A_1,\dots,A_M\}$ of $Q$
and a selection map $\alpha$
such that
\[
\varphi(T_s,A_i,u_i)\subset Q,
\qquad i=1,\dots,M.
\]

A sampled-data controller that enforces invariance
can therefore be constructed as follows:
at time $t_k$, determine an index
$i\in\{1,\dots,M\}$ such that
$x(t_k)\in A_i$ and apply $u_i$
over $[t_k,t_{k+1})$.
Since each $u_i$ keeps all states in $A_i$
inside $Q$ over one sampling period,
recursive application guarantees \eqref{eq:invariance}.
Thus invariance reduces to selecting
an index $i\in\{1,\dots,M\}$ at each sampling time.
Controller design becomes equivalent to
constructing a partition of $Q$
together with an assignment of each partition element
to an open-loop control in $\mathcal U_{T_s}$.
The geometry of the sets $A_i$ is not prescribed,
and determining a partition achieving minimal
cardinality $r_{\mathrm{inv}}(T_s,Q)$
is, in general, nontrivial.
\subsection{Learning a Minimal Control Codebook via VQ-AE}
A constructive approximation of the partition–assignment pair
is obtained using a vector-quantized autoencoder (VQ-AE).
A discrete latent space indexing
$\mathcal Z = \{1,\dots,M\}$ of fixed cardinality $M$
parameterizes the partition of $Q$.
To generate this partition, an encoder network $f_\theta: \mathbb{R}^n \to \mathbb{R}^d$, parametrized by $\theta$, maps each state $x(t_k)$ to a latent vector. Our discrete encoder $E_\theta : \mathbb{R}^n \to \mathcal{Z}$ then assigns each state to the nearest code index in a codebook $\{e_1, \dots, e_M\} \subset \mathbb{R}^d$:
\[
E_\theta(x) = \arg\min_{k \in \mathcal{Z}} \left\| f_\theta(x) - e_k \right\|_2^2.
\]
A decoder $D_\phi : \mathcal{Z} \to \mathcal{U}_{T_s}$, parametrized by $\phi$, maps each code index to a control input over the horizon $T_s$.
Therefore, applied control over $[t_k,t_{k+1})$ is
$
u(\cdot) = D_\phi(E_\theta(x(t_k))).
$
Let $\{x_j\}_{j=1}^N \subset Q$ be sampled states.
For each $x_j$, rollout over one sampling period yields
$\varphi(t,x_j,D_\phi(E_\theta(x_j)))$.
Empirical invariance violation is measured by
\[
\mathcal L_{\mathrm{inv}}
=
\frac{1}{N}
\sum_{j=1}^N
\int_0^{T_s}
\bigl[
\operatorname{dist}
(\varphi(t,x_j,D_\phi(E_\theta(x_j))),Q)
\bigr]_+ dt,
\]
where $[\cdot]_+$ denotes the positive part denoting a minimum distance from a state to the set $Q$. To train the VQ-AE end-to-end, gradients are passed through
the discrete assignment via the straight-through estimator. The codebook vectors are updated using exponential moving averages (EMA) of the encoder outputs assigned to each code, which stabilizes training without requiring explicit codebook
gradients found in standard VQ-AE formulations. A commitment loss
\[
\mathcal{L}_{\mathrm{commit}} = \frac{1}{N} \sum_{j=1}^N \| z_j - \operatorname{sg}(e_{E_\theta(x_j)}) \|_2^2,
\]
where $z_j = f_\theta(x_j)$ and $\operatorname{sg}(\cdot)$ denotes
the stop-gradient operator. This term anchors the encoder outputs to
their assigned codebook vectors.
Consistent with minimizing the number of open-loop
controls, the soft assignment
\[
p_k(x_j) = \frac{\exp\left( -\| z_j - e_k \|_2^2 / \tau \right)}{\sum_{m=1}^M \exp\left( -\| z_j - e_m \|_2^2 / \tau \right)},
\]
where $\tau > 0$ controls the temperature, gives empirical usage $\pi_i = \frac{1}{N} \sum_{j=1}^N p_i(x_j)$ and entropy
$
H(\pi) = -\sum_{i=1}^M \pi_i \log \pi_i.
$
The codebook usage loss is then defined as the weighted entropy:
\[
\mathcal{L}_{\mathrm{codebook}} =  \left(-\sum_{i=1}^M w_i \, \pi_i \log \pi_i\right)^p,
\]
where $p > 1$ concentrates compression pressure on large
codebooks and softens it near collapse. The per-code weights
$w_i$, computed via a sigmoid on each code's empirical safety
rate, shield codes with high violation rates from elimination.
To prevent excessive control effort, a regularization loss
$\mathcal{L}_{\mathrm{reg}}$ penalizes the magnitude of the
learned control sequences. The complete training objective is
\[
\min_{\theta, \phi} \quad \lambda_{\mathrm{inv}} \mathcal{L}_{\mathrm{inv}} + \lambda_{\mathrm{reg}} \mathcal{L}_{\mathrm{reg}} + \lambda_{\mathrm{vq}} \mathcal{L}_{\mathrm{commit}} + \lambda_{pr}\mathcal{L}_{\mathrm{codebook}}.
\]
Training states are drawn from a boundary-biased distribution over $Q$, where a fixed fraction of samples are concentrated near $\partial Q$ and the remainder drawn uniformly from the interior. This overweighting of boundary regions ensures sufficient coverage where invariance violations are most likely to occur. The learned decoder outputs $\{D_\phi(i)\}_{i=1}^M$ which is an upper bound on $r_{\mathrm{inv}}(T_s, Q)$, with tightness governed by the richness of the admissible control and partition sets, while the encoder induces the associated partition of $Q$. Since the space of open-loop controls on $[0,T_s]$ is
infinite-dimensional, we consider the decoder outputs to be piecewise-constant
controls with $H$ segments of duration $\bar T_s$
satisfying $H\bar T_s = T_s$. The resulting controller becomes
\begin{equation}
\label{eq:controller}
\begin{aligned}
\mathcal C(x(t_k)) &= u_{E_\theta(x(t_k))},\qquad
t\in[t_k+j\bar T_s,\,t_k+(j+1)\bar T_s),\\
A_i &= \{x\in Q : E_\theta(x)=i\}, \qquad i=1,\dots,M,\\
u_i &= D_\phi(i)=\{u_i^j\}_{j=0}^{H-1}.
\end{aligned}
\end{equation}
where $T_s=H\bar T_s$ and each $u_i^j$ is constant.

\section{Correctness of the Learned Controller}
Each index $i \in \{1,\dots,M\}$ is associated with a fixed
open-loop control sequence $u_i = D_\phi(i) \in \mathcal U_{T_s}$,
yielding a control family $\mathcal V_{T_s}$. The main difficulty in verification lies in the geometry of the
state-space partition
$
A_i := \{x \in Q : E_\theta(x)=i\}.
$
For nonlinear encoders these regions may have implicit and
nonconvex boundaries, making explicit characterization difficult. To obtain an analytic description of the partition, we define the encoder
to be an invertible affine map $T(x)=Wx+b$, where
$W\in\mathbb{R}^{n\times n}$ and $\det W \neq 0$. We use the matrix exponential parametrization to enforce an invertible mapping~\cite{casado2019orthogonal}. Nearest-neighbor code selection in latent space induces
Voronoi cells in the latent coordinates, which are convex polytopes.
Since each region is the preimage $A_i = T^{-1}(V_i)$ of a Voronoi
cell and $T$ is affine, each $A_i$ is a convex polytope, i.e., a
finite intersection of halfspaces $\{\alpha_j^\top x \leq \beta_j\}$.

The set $Q$ introduced in the problem formulation is chosen as a
quadratic sublevel set
\[
Q = \{x \in \mathbb{R}^n : V(x) \le \rho\}, \qquad
V(x) = x^\top P x,
\]
with $P \succ 0$. The matrix $P$ is obtained from the LQR design
associated with the linearization of \eqref{eq:plant} around the
equilibrium. For nonlinear systems, LQR controllers typically stabilize the dynamics
locally around the equilibrium, and therefore the corresponding Lyapunov
sublevel sets provide natural candidates for invariant regions. The set
$Q$ is chosen within this neighborhood. This choice is justified a
posteriori, as we later verify that the learned controller preserves
invariance of $Q$.

Verification proceeds region-wise over the learned partition.
Each region $A_i$ is associated with a control sequence
$\{u_i^j\}_{j=0}^{H-1}$ applied over intervals of length $\bar T_s$.
Initialization of $A_i \cap Q$ occurs in Line~\ref{alg:init}, and the centroid
$c$ is computed in Line~\ref{alg:centroid}. The Gronwall-based enclosure
requires a Lipschitz constant valid over the entire reachable tube, not just
the initial region, since trajectories may leave $A$ over $[0,\bar T_s]$.
\textsc{ComputeL} addresses this by iteratively expanding the enclosure:
at each iteration, $L$ defines a reachable set $B$ via rollouts and a
polyhedral over-approximation $\tilde A$, over which a Lipschitz estimate
$\hat L$ is recomputed. The bound is updated as
$L \leftarrow \max(L+\varepsilon,\hat L)$ and the process terminates when
$\hat L \le L$, ensuring validity over a region containing the reachable tube.
Since $f$ is globally Lipschitz, $\hat L$ is bounded while $L$ grows by at
least $\varepsilon$, so termination is finite. The enclosure $\tilde A$ is
intersected with $Q$ in Line~\ref{alg:intersect}, after which an SOS
feasibility problem is solved in Line~\ref{alg:sos}.

\subsection{Sum-of-Squares Certification}

For each region $A_i = \{x : a_{il}^\top x \le \beta_{il},\ l=1,\dots,m_i\}$ and control segment $u_i^j$, invariance is certified by enforcing $\dot V(x) \le 0$ on $A_i \cap Q$ via the S-procedure ~\cite{parrilo2003semidefinite}. SOS verification requires $\nabla V(x)^\top f(x, u_i^j)$ to be polynomial in $x$; since the true dynamics contain trigonometric nonlinearities, we replace $f$ with a polynomial surrogate $\tilde{f}$ obtained by 3rd-order Taylor expansion around the equilibrium, justified by the choice of $Q$ as a near-equilibrium set with small attitude deviations. Define $g_l(x) = \beta_{il} - a_{il}^\top x$ and $g_Q(x) = \rho - x^\top Px$. The condition $\dot V(x) = \nabla V(x)^\top \tilde{f}(x, u_i^j) \le 0$ on $A_i \cap Q$ is enforced by requiring
\[
-\dot V(x)
- \sum_{l=1}^{m_i} \lambda_{il}(x)\, g_l(x)
- \lambda_Q(x)\, g_Q(x)
\]
to be a sum-of-squares polynomial, with SOS multipliers $\lambda_{il}(x)$, $\lambda_Q(x)$. Feasibility certifies that surrogate-model trajectories starting in $A_i \cap Q$ remain inside $Q$ under $u_i^j$.

\begin{algorithm}[t]
\caption{\textsc{ComputeL}}
\begin{algorithmic}[1]
\State \textbf{Input:} $A$, $c$, $u$, $\bar T_s$, $\varepsilon$
\State $\tilde A \gets A$, \quad $L \gets 0$
\Loop
    \State $\hat L \gets \operatorname{Lip}(f(\cdot,u),\tilde A)$
    \State \textbf{if} $\hat L \le L$ \textbf{then return} $(L,\tilde A)$
    \State $L \gets \max(L+\varepsilon,\hat L)$
    \State $B \gets \bigcup_{t\in[0,\bar T_s]} \left\{\varphi(t,c,u) + e^{L\bar T_s}(x-c) : x\in A\right\}$
    \State $\beta_l \gets \max_{y\in B} a_l^\top y,\quad l=1,\dots,m$
    \State $\tilde A \gets \{x : a_l^\top x \le \beta_l,\ l=1,\dots,m\}$
\EndLoop
\end{algorithmic}
\end{algorithm}
\begin{algorithm}[t]
\caption{Iterative Forward Certification}
\begin{algorithmic}[1]
\State \textbf{Input:} $\{(A_i,\{u_i^j\}_{j=0}^{H-1})\}_{i=1}^M$, $Q$, $\bar T_s$, $\varepsilon$
\For{$i=1,\dots,M$}
    \State \label{alg:init} $A \gets A_i \cap Q$
    \For{$j=0,\dots,H-1$}
        \State \label{alg:centroid} $c \gets \operatorname{centroid}(A)$
        \State \label{alg:computeL} $(L, \tilde A) \gets \textsc{ComputeL}(A,c,u_i^j,\bar T_s,\varepsilon)$
        \State \label{alg:intersect} $A \gets \tilde A \cap Q$
        \State \label{alg:sos} \textbf{if} SOS$(A,u_i^j)$ is infeasible \textbf{then return} \textsc{Fail}
    \EndFor
\EndFor
\State \Return \textsc{Success}
\end{algorithmic}
\end{algorithm}
\begin{theorem}
\label{lem:certified_controller}
Suppose there exist $H \in \mathbb{N}_{>0}$, $\bar T_s > 0$, and  $Q \subset \mathbb{R}^n$
such that Algorithm~2 completes, then the closed-loop system
$
\dot x(t) = \tilde{f}\bigl(x(t), \mathcal C(x(t))\bigr),
$
where $\mathcal C$ is the constructed controller defined in~\eqref{eq:controller},
is forward invariant on $Q$.
\end{theorem}
\begin{proof}
Assume Algorithm~2 completes. We show that under each control segment, the Lipschitz-based reachable-set enclosure combined with the SOS certificate confines trajectories to $Q$, even though the polyhedral overapproximation $\tilde A$ may extend beyond $Q$.

Consider a region $A\subset Q$ at the beginning of a segment of length $\bar T_s$ under control $u_i^j$, and let $c=\operatorname{centroid}(A)$. Let $(L,\tilde A)$ denote the output of \textsc{ComputeL}. The sequence $\{L_k\}$ is non-decreasing. Since $f$ is globally Lipschitz, $\hat L$ is bounded, while the update $L \leftarrow \max(L+\varepsilon,\hat L)$ ensures $L$ increases by at least $\varepsilon$ at each unsuccessful iteration, so the loop terminates finitely. Upon termination, $\hat L \le L$ ensures $\mathrm{Lip}(f(\cdot, u_i^j), \tilde{A}) \le L$.

Since $L$ is a valid Lipschitz constant on $\tilde{A}$, the Bellman Gronwall inequality gives, for any $x\in A$,
\[
\begin{aligned}
    \|\varphi(t,x,u_i^j)-\varphi(t,c,u_i^j)\|& \le e^{Lt}\|x-c\|\\
    & \le e^{L\bar T_s}\|x-c\|,\ \forall\, t\in[0,\bar T_s].
\end{aligned}
\]
Hence $\varphi(t,A,u_i^j)\subseteq B$ for all $t\in[0,\bar T_s]$, where
\[
B=\bigcup_{t\in[0,\bar T_s]} \bigl\{\varphi(t,c,u_i^j) + e^{L\bar T_s}(x-c):x\in A\bigr\}.
\]
The polyhedral enclosure $\tilde A=\{x: a_{il}^\top x\le \beta_{il},\; l=1,\dots,m_i\}$ with $\beta_{il}=\max_{y\in B} a_{il}^\top y$ satisfies $B\subseteq\tilde A$ by construction, so $\varphi([0,\bar T_s],A,u_i^j)\subseteq B\subseteq\tilde A$.

The region is updated as $A\leftarrow \tilde A\cap Q$, and feasibility of the SOS program guarantees $\dot V(x)\le0$ for all $x\in \tilde A\cap Q$. It remains to show that this certificate, established only on $\tilde A\cap Q$, suffices to prevent trajectories from reaching the portion of $\tilde A$ outside $Q$. Suppose for contradiction that some $x_0\in A\cap Q$ satisfies $V(\varphi(t^*,x_0,u_i^j))>\rho$ for some $t^*\in(0,\bar T_s]$. Let $t_e=\inf\{t>0:V(\varphi(t,x_0,u_i^j))>\rho\}$. For all $t\in[0,t_e]$ the trajectory satisfies $V(\varphi(t,x_0,u_i^j))\le\rho$, hence lies in $\tilde A\cap Q$. The SOS certificate gives $\dot V\le 0$ along the trajectory, so $V(\varphi(t_e,x_0,u_i^j))\le V(x_0)\le\rho$, contradicting the definition of $t_e$. Therefore any trajectory starting in $A\cap Q$ remains in $Q$ over $[0,\bar T_s]$.

The argument extends to all segments by forward chaining over $j = 0, \dots, H-1$: at the start of each segment the current region satisfies $A \subset Q$, and the preceding argument shows that trajectories remain in $Q$ and lie within $\tilde{A} \cap Q$, which Algorithm~2 passes as the input region to the next segment (Line~8). Thus trajectories starting in $A_i$ remain in $Q$ over $[t_k, t_{k+1})$. Since $\{A_i\}_{i=1}^M$ partitions $Q$, every $x(0) \in Q$ belongs to some region, and the closed-loop system is forward invariant on $Q$.
\end{proof}

\section{Empirical Results}
We consider the nonlinear quadrotor \cite{mahony2012multirotor} with state $x = [p_x, p_y, p_z, v_x, v_y, v_z, \phi, \theta, \psi, \omega_x, \omega_y, \omega_z]^\top \in \mathbb{R}^{12}$ consisting of position, linear velocity, Euler angles, and body angular rates, and control input $u = [T, \tau_x, \tau_y, \tau_z]^\top \in \mathbb{R}^4$, where $T$ is thrust and $\tau$ is torque. The dynamics $\dot{x} = f(x,u)$ are
\begin{equation}
\label{eq:quadrotor_dynamics}
\begin{aligned}
\dot{p} &= v, \quad
\dot{v} = \tfrac{T}{m} R(\phi,\theta,\psi)\, e_3 - g e_3, \\
\dot{\eta} &= W(\phi,\theta)\, \omega, \quad
\dot{\omega} = I^{-1}\!\left( -\omega \times I\omega + \tau \right).
\end{aligned}
\end{equation}
where $m$ is the vehicle mass, $g$ is gravitational acceleration, $I$ is the inertia matrix, $e_3 = [0,0,1]^\top$ is the vertical unit vector, $R(\phi,\theta,\psi) \in SO(3)$ is the rotation matrix and $W(\phi,\theta)$ is the Euler angle kinematic matrix. The training process and closed-loop simulations use the polynomial surrogate $\tilde{f}$, so that the reported invariance rates are directly comparable to the IFC certificate established in Theorem~\ref{lem:certified_controller}. The sublevel parameter is set to $\rho = 5$, the sampling
period is $T_s = H \bar{T}_s$ with $\bar{T}_s = 0.01$\,s and
$H = 2$, and the VQ-AE is initialized with $M = 8192$ codes.

We justify the polynomial surrogate $\tilde{f}$ by evaluating
$|V(\varphi_f(t)) - V(\varphi_{\tilde{f}}(t))|$ at every integration
step over $[0, T_s]$ with $T_s = 0.02$\,s across $10^7$ trajectories
from states sampled uniformly in $Q$ under admissible controls. The
99th-percentile gap is $3.0 \times 10^{-4}$ ($0.006\%$ of $\rho$),
with a maximum of $7.0 \times 10^{-3}$ ($0.14\%$ of $\rho$),
indicating that the surrogate approximation error is negligible.

Closed-loop invariance is evaluated by sampling $10^5$ initial states uniformly from $Q$, simulating the closed-loop system for 10\,s, and recording a trajectory as a failure if it leaves $Q$ at any point. After training, rarely used codes are pruned and states are reassigned to their nearest surviving code, reducing the number of regions that must be verified. We investigate three questions: (i)~how far can the codebook
be compressed before invariance breaks down
(Section~\ref{sec:compression}), (ii)~how does the learned
controller compare to baselines
(Section~\ref{sec:baselines}), and (iii)~what is the coarsest
image resolution at which a vision-based observer can still
support the controller (Section~\ref{sec:resolution}).
\begin{figure}[ht]
    \centering
    \includegraphics[width=0.75\linewidth]{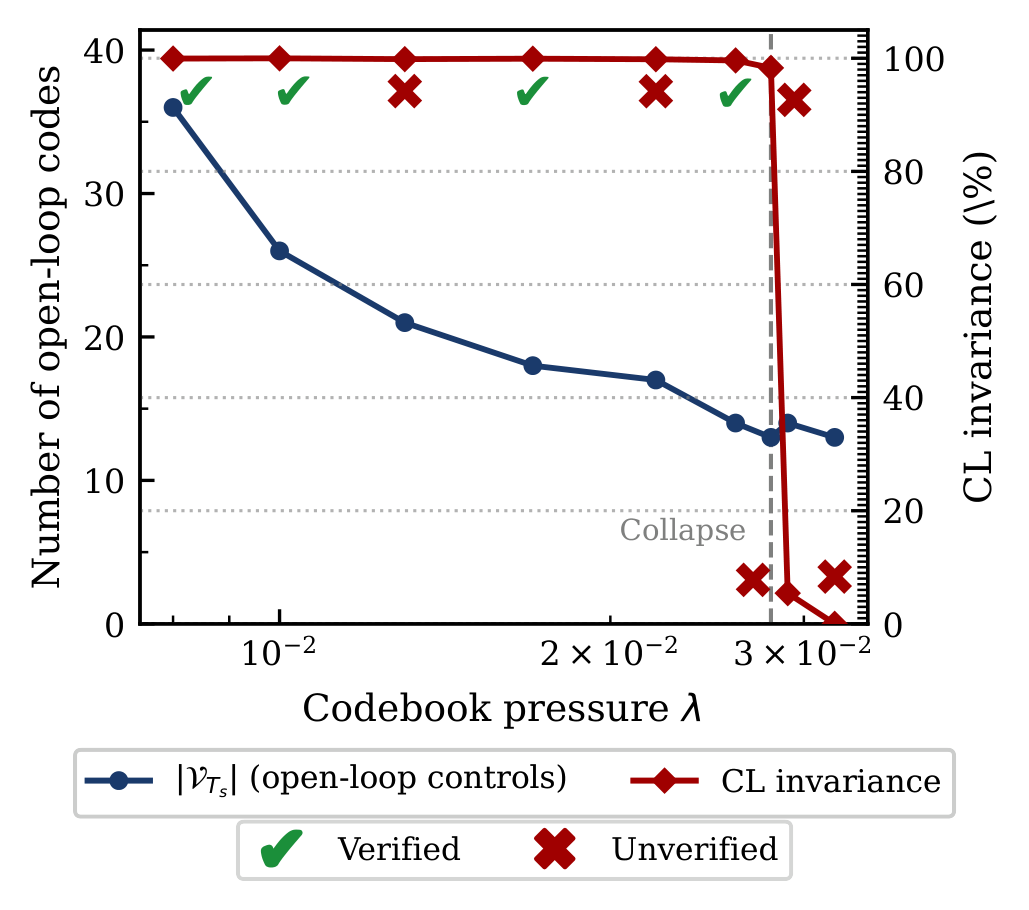}
    \caption{Codebook size $|\mathcal{V}_{T_s}|$ and closed-loop invariance vs.\ codebook pressure $\lambda_{pr}$. Invariance holds at $100\%$ down to $|\mathcal{V}_{T_s}|=14$ codes, then collapses. Green checkmarks indicate IFC certification success; all controllers at or beyond the collapse point are unverified.}
    \label{fig:lambda_sweep}
\end{figure}
\subsection{Codebook Compression}\label{sec:compression} We stress-test the pipeline by varying the entropy pressure $\lambda_{pr}$ that penalizes codebook usage. Fig.~\ref{fig:lambda_sweep} shows the result: as $\lambda_{pr}$ increases, the number of active open-loop controls $|\mathcal{V}_{T_s}|$ decreases monotonically from 36 to 14 while closed-loop invariance remains near $100\%$. Below 14 codes, invariance collapses to below $10\%$, suggesting that the invariance task for this system and choice of $Q$ requires at least $\sim\!14$ distinct open-loop signals. This empirical bound has the inductive bias of the architecture; the true $r_{\mathrm{inv}}(T_s, Q)$ may be lower under a richer control set. We also observe that even at higher code counts, a small number of closed-loop trajectories violate invariance in simulation (up to $0.01\%$), while the IFC pipeline successfully certifies these controllers. We attribute this gap to accumulating integration error of the RK4 scheme; the open-loop invariance violation rate remains below $3.6 \times 10^{-5}\%$. IFC certification is consistent with the observed safety trend. All controllers 
at or beyond the collapse point ($\lambda \geq 0.028$) are unverified, consistent 
with empirical invariance degradation. Among the remaining controllers, those with 
the most codes tend to be certified, while some intermediate cases are not. Since 
IFC provides a sufficient condition for invariance via reachable-set overapproximation, 
unverified controllers are not necessarily unsafe; certification failure indicates 
only that the overapproximation was conservative to confirm invariance.

\begin{table}[ht]
\centering
\caption{Comparison of codebook size and closed-loop invariance rate across baselines, evaluated over $10^5$ trials for 10\,s.}
\label{tab:baselines}
\small
\begin{tabular}{|l|r|r|}
\hline
Method & $|\mathcal{V}_{T_s}|$ & CL Inv.\ (\%) \\
\hline
VQ-AE (ours, pruned)   &  26 & 99.979 \\
\hline
VQ-AE (ours, unpruned) & 171 & 99.989 \\
\hline
VQ-AE partition + OL LQR  & 171 & 0.000 \\
\hline
Uniform grid + NN       & 4096 & 99.975 \\
\hline
Uniform grid + CL LQR & 2.2M & 76.870 \\
\hline
\end{tabular}
\end{table}

\subsection{Baseline Comparisons}\label{sec:baselines}

We compare the proposed controller, trained with codebook pressure $\lambda=0.01$ and evaluated over $10^5$ initial conditions, against four alternatives. Methods (i) and (ii) correspond to our approach: (i)~\emph{VQ-AE (pruned)} and (ii)~\emph{VQ-AE (unpruned)}, where the pruned version discards rarely used codes and reassigns states to the nearest remaining code. The baselines are (iii)~\emph{VQ-AE partition + OL LQR}, which retains the learned partition but replaces the learned decoder outputs with a open-loop LQR rollout for $T_s = 0.02\,s$ initialized at the region centroid; (iv)~\emph{uniform grid + NN}, which partitions each of the 12 state dimensions into 2 intervals ($2^{12}=4096$ cells) and assigns each cell an open-loop control via an independently trained neural network; and (v)~\emph{uniform grid + CL LQR}, which partitions each dimension into 5 intervals ($5^{12}\approx 244$M cells), retains only the ${\sim}2.2$M cells intersecting $Q$, and applies closed-loop LQR with state re-sampling at each $\bar{T}_s$ interval, in contrast to the open-loop rollout used in (iii).

All methods are evaluated over $10^5$ initial conditions sampled uniformly from $Q$, with 10\,s simulations. Table~\ref{tab:baselines} shows that (ii) slightly outperforms (iii) while using approximately $24\times$ fewer codes (171 vs.\ 4096). After pruning, (i) uses 26 codes, achieving a $157\times$ reduction relative to (iv) while still maintaining higher closed-loop invariance. Method (v), despite using over $2.2$ million partitions, achieves substantially lower invariance. We attribute this to quantization error: even with ${\sim}2.2$M cells, the within-cell state variation exceeds the tolerance of the LQR controller over each $\bar{T}_s$ interval, particularly near $\partial Q$. Remarkably, neither the partition nor the controller alone is sufficient: (iii) fails despite using the learned partition, and (v) underperforms despite using a dense partition with closed-loop control. In contrast, the joint encoder--decoder structure in (i) achieves higher invariance with nearly $10^5\times$ fewer regions, highlighting the importance of co-designing the partition and control signals.

\subsection{Sensing Resolution and Invariance}\label{sec:resolution}
Since $\mathcal{C}$ maps states to discrete indices $i \in \{1,\dots,M\}$,
the observer need not reconstruct the full state; it suffices to identify
which region $A_i$ the current state belongs to. This raises a natural
question: what is the coarsest image resolution at which the observer can
still support invariance?

We fix the same controller used for baselines comparison and the observer
network $\mathcal{O}_\psi$ is trained to estimate the state from $r \times r$
grayscale images of an ArUco marker~\cite{garrido2014automatic} rendered at
the quadrotor's current pose. At each sampling time $t_k$, the observer processes the $K$-step observation history with $K=4$
$\mathbf{y}_k := (y_{k-K}(r), \ldots, y_k(r))$,
where $y_k(r) := h(x(t_k), r) \in \mathbb{R}^{r \times r}$,
and produces a state estimate $\hat{x}_k = \mathcal{O}_\psi(\mathbf{y}_k)$
passed to $\mathcal{C}$ to yield the composite mapping
$g = \mathcal{C} \circ \mathcal{O}_\psi$.

\begin{figure}[ht]
    \centering
    \includegraphics[width=0.7\linewidth]{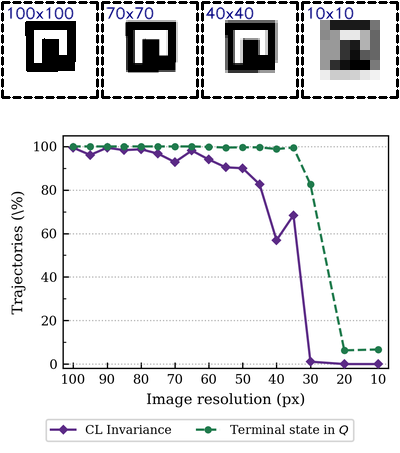}
    \caption{Outputs at different resolutions (indicated above). Closed-loop invariance and terminal safety over 10\,s versus image resolution. Invariance remains above $99\%$ down to $\sim65$\,px, then degrades with lower resolution.}
    \label{fig:resolution_sweep}
\end{figure}

The observer uses a ResNet-18 \cite{targ2016resnet} backbone applied independently to each image
in the history, followed by an MLP head that fuses per-image features into
a state estimate. Only $\psi$ is optimized; $E_\theta$ and $D_\phi$ are
frozen. Training combines a mean squared state reconstruction loss with a binary cross-entropy loss that penalizes region misidentification. Correct code selection matters because it ensures the true state lies in a partition region with a certified control, maintaining invariance even when the state estimate is noisy. The reconstruction loss is necessary alongside classification: without spatial grounding, the classification loss alone is fragile and cannot distinguish between misclassifying to a neighboring region versus a distant one, so the observer fails to learn the geometric structure of the partition. For each candidate resolution $r$, the observer is
independently retrained and the closed-loop system is evaluated over 10\,s
from 1000 initial states sampled uniformly from $Q$.

Fig.~\ref{fig:resolution_sweep} reports two metrics as resolution decreases:
the fraction of trajectories that never leave $Q$ (closed-loop invariance),
and the fraction that end inside $Q$ at the terminal time of each open-loop
segment (terminal safety). Both metrics are reported because vision-based
control is susceptible to transient misclassifications that cause brief
departures from $Q$ without diverging. Closed-loop invariance remains near $99\%$ for resolutions down to approximately $r=65$\,px, then degrades progressively. At $r=40$\,px roughly $58\%$ of trajectories maintain invariance; the controller fails almost entirely below $r=30$\,px. The terminal safety curve is substantially more robust, staying above $99\%$ down to approximately $r=30$\,px, indicating that trajectories which leave $Q$ frequently return by the end of the open-loop control execution.

\section{Conclusion}
This paper studied the problem of constructing low-complexity controllers that certify forward invariance of a compact set under sampled-data control. A VQ-AE architecture learns a state-space partition and control codebook, and an iterative forward certification algorithm verifies correctness via Lipschitz-based reachable-set enclosures and SOS feasibility. On a 12-dimensional nonlinear quadrotor, the learned controller maintains invariance with 14 control signals, a $157\times$ reduction over a uniform grid baseline, and vision-based experiments show that closed-loop invariance is preserved down to approximately $65$\,px resolution.

Future work includes generalizing the verification framework from polynomial surrogate dynamics to broader model classes, extending the encoder beyond affine maps while preserving certifiable partition geometry, and establishing theoretical lower bounds on codebook size.
\balance
\bibliography{software}

@article{colonius2012stabilization,
    author = {Colonius, Fritz},
    title = {Minimal Bit Rates and Entropy for Exponential Stabilization},
    journal = {SIAM Journal on Control and Optimization},
    volume = {50},
    number = {5},
    pages = {2988-3010},
    year = {2012},
    doi = {10.1137/110829271}
}

@article{colonius2009invariance,
  title={Invariance entropy for control systems},
  author={Colonius, Fritz and Kawan, Christoph},
  journal={SIAM Journal on Control and Optimization},
  volume={48},
  number={3},
  pages={1701--1721},
  year={2009},
  publisher={SIAM}
}

@article{kawan2013invariance,
  title={Invariance entropy for deterministic control systems},
  author={Kawan, Christoph},
  journal={Lecture notes in mathematics},
  volume={2089},
  year={2013},
  publisher={Springer}
}

@article{mahony2012multirotor,
  title={Multirotor aerial vehicles: Modeling, estimation, and control of quadrotor},
  author={Mahony, Robert and Kumar, Vijay and Corke, Peter},
  journal={IEEE robotics \& automation magazine},
  volume={19},
  number={3},
  pages={20--32},
  year={2012},
  publisher={IEEE}
}

@article{han2015learning,
  title={Learning both weights and connections for efficient neural network},
  author={Han, Song and Pool, Jeff and Tran, John and Dally, William},
  journal={Advances in neural information processing systems},
  volume={28},
  year={2015}
}

@article{frankle2018lottery,
  title={The lottery ticket hypothesis: Finding sparse, trainable neural networks},
  author={Frankle, Jonathan and Carbin, Michael},
  journal={arXiv preprint arXiv:1803.03635},
  year={2018}
}

@inproceedings{atanov2024far,
  title={How Far Can a 1-Pixel Camera Go? Solving Vision Tasks Using Photoreceptors and Computationally Designed Visual Morphology},
  author={Atanov, Andrei and Fu, Jiawei and Singh, Rishubh and Yu, Isabella and Spielberg, Andrew and Zamir, Amir},
  booktitle={European Conference on Computer Vision},
  pages={458--476},
  year={2024},
  organization={Springer}
}

@inproceedings{klotz2024minimalist,
  title={Minimalist vision with freeform pixels},
  author={Klotz, Jeremy and Nayar, Shree K},
  booktitle={European Conference on Computer Vision},
  pages={329--346},
  year={2024},
  organization={Springer}
}

@article{van2017neural,
  title={Neural discrete representation learning},
  author={Van Den Oord, Aaron and Vinyals, Oriol and others},
  journal={Advances in neural information processing systems},
  volume={30},
  year={2017}
}

@article{brunke2022safe,
  title={Safe learning in robotics: From learning-based control to safe reinforcement learning},
  author={Brunke, Lukas and Greeff, Melissa and Hall, Adam W and Yuan, Zhaocong and Zhou, Siqi and Panerati, Jacopo and Schoellig, Angela P},
  journal={Annual Review of Control, Robotics, and Autonomous Systems},
  volume={5},
  pages={411--444},
  year={2022},
  publisher={Annual Reviews}
}

@book{beard2012small,
  title={Small unmanned aircraft: Theory and practice},
  author={Beard, Randal W and McLain, Timothy W},
  year={2012},
  publisher={Princeton university press}
}

@article{zhang2024end,
  title={End-to-end nano-drone obstacle avoidance for indoor exploration},
  author={Zhang, Ning and Nex, Francesco and Vosselman, George and Kerle, Norman},
  journal={Drones},
  volume={8},
  number={2},
  pages={33},
  year={2024},
  publisher={MDPI}
}

@article{bengio2013estimating,
  title={Estimating or propagating gradients through stochastic neurons for conditional computation},
  author={Bengio, Yoshua and L{\'e}onard, Nicholas and Courville, Aaron},
  journal={arXiv preprint arXiv:1308.3432},
  year={2013}
}

@inproceedings{wang2026theoretically,
  title={A Theoretically-Grounded Codebook for Digital Semantic Communications},
  author={Wang, Lingyi and Shelim, Rashed and Saad, Walid and Ramakrishnan, Naren},
  booktitle={2026 IEEE 23rd Consumer Communications \& Networking Conference (CCNC)},
  pages={1--6},
  year={2026},
  organization={IEEE}
}

@article{parrilo2003semidefinite,
  title={Semidefinite programming relaxations for semialgebraic problems},
  author={Parrilo, Pablo A},
  journal={Mathematical programming},
  volume={96},
  number={2},
  pages={293--320},
  year={2003},
  publisher={Springer}
}

@article{garrido2014automatic,
  title={Automatic generation and detection of highly 
         reliable fiducial markers under occlusion},
  author={Garrido-Jurado, S. and Mu{\~n}oz-Salinas, R. 
          and Madrid-Cuevas, F.J. and Mar{\'\i}n-Jim{\'e}nez, M.J.},
  journal={Pattern Recognition},
  volume={47},
  number={6},
  pages={2280--2292},
  year={2014},
  publisher={Elsevier}
}

@article{falanga2020dynamic,
  title={Dynamic obstacle avoidance for quadrotors with event cameras},
  author={Falanga, Davide and Kleber, Kevin and Scaramuzza, Davide},
  journal={Science Robotics},
  volume={5},
  number={40},
  pages={eaaz9712},
  year={2020},
  publisher={American Association for the Advancement of Science}
}

@article{targ2016resnet,
  title={Resnet in resnet: Generalizing residual architectures},
  author={Targ, Sasha and Almeida, Diogo and Lyman, Kevin},
  journal={arXiv preprint arXiv:1603.08029},
  year={2016}
}

@inproceedings{liberzon2016entropy,
  title={Entropy and minimal data rates for state estimation and model detection},
  author={Liberzon, Daniel and Mitra, Sayan},
  booktitle={Proceedings of the 19th international conference on hybrid systems: Computation and control},
  pages={247--256},
  year={2016}
}

@inproceedings{sibai2018state,
  title={State estimation of dynamical systems with unknown inputs: Entropy and bit rates},
  author={Sibai, Hussein and Mitra, Sayan},
  booktitle={Proceedings of the 21st International Conference on Hybrid Systems: Computation and Control (part of CPS Week)},
  pages={217--226},
  year={2018}
}

@article{casado2019orthogonal,
  author       = {Mario Lezcano Casado and
                  David Mart{\'{\i}}nez{-}Rubio},
  title        = {Cheap Orthogonal Constraints in Neural Networks: {A} Simple Parametrization
                  of the Orthogonal and Unitary Group},
  journal      = {CoRR},
  volume       = {abs/1901.08428},
  year         = {2019},
  url          = {http://arxiv.org/abs/1901.08428},
  eprinttype   = {arXiv},
  eprint       = {1901.08428},
  bibsource    = {dblp computer science bibliography, https://dblp.org}
}

@inproceedings{hespanha2002towards,
  title={Towards the control of linear systems with minimum bit-rate},
  author={Hespanha, Joao and Ortega, Antonio and Vasudevan, Lavanya},
  booktitle={Proc. 15th Int. Symp. on Mathematical Theory of Networks and Systems (MTNS)},
  pages={1--15},
  year={2002}
}
\end{document}